\documentclass[letterpaper, 10pt, conference]{ieeeconf}

\usepackage{graphicx}
\usepackage{amsmath}
\usepackage{amssymb}
\usepackage{subfig}
\usepackage{bm}
\usepackage{comment}

\IEEEoverridecommandlockouts
\def\ket#1{|#1 \rangle}
\def\RR{\mathbb{R}}
\def\H{\mathcal{H}}
\def\vec#1{\mathbf{#1}}
\def\A{\mathtt{A}}
\def\S{\mathtt{S}}
\def\I{\mathtt{I}}
\def\ld{\mathop{\rm ld}}
\def\Tr{\mathrm{Tr}}
\def\vv{\bm v}
\newtheorem{theorem}{Theorem}

\title{\LARGE\bf Robustness of Quantum Systems Subject to Decoherence: \\ Structured Singular Value Analysis?}

\author{S.\,G.\ Schirmer \and F.\,C.\ Langbein \and C.\,A.\ Weidner \and E.\,A.\ Jonckheere%
\thanks{S.\,G.\ Schirmer is with the Faculty of Science \& Engineering, Swansea University, Singleton Park, Swansea, SA2 8PP, UK, {\tt\small lw1660@gmail.com}.}%
\thanks{F.\,C.\ Langbein is with the School of Computer Science and Informatics, Cardiff University, UK, {\tt\small frank@langbein.org}.}%
\thanks{C.\,A.\ Weidner is with the Institute for Physics and Astronomy at Aarhus University, Denmark, {\tt\small cweidner@phys.au.dk}.}%
\thanks{E.\,A.\ Jonckheere is with the Department of Electrical Engineering, University of Southern California, Los Angeles, CA 90089, {\tt jonckhee@usc.edu}.}
\thanks{This work is partially supported by NSF IRES 1829078.}}

\begin{document}

\maketitle
\thispagestyle{empty}
\pagestyle{empty}

\begin{abstract}
We study the problem of robust performance of quantum systems under structured uncertainties.  A specific feature of closed (Hamiltonian) quantum systems is that their poles lie on the imaginary axis and that neither a coherent controller nor physically relevant structured uncertainties can alter this situation.  This changes for open systems where decoherence ensures asymptotic stability and creates a unique landscape of pure performance robustness, with the distinctive feature that closed-loop stability is secured by the underlying physics and needs not be enforced.  This stability, however, is often detrimental to quantum-enhanced performance, and additive perturbations of the Hamiltonian give rise to dynamic generators that are nonlinear in the perturbed parameters, invalidating classical paradigms to assess robustness to structured perturbations such as singular value analysis. This problem is addressed using a fixed-point iteration approach to determine a maximum perturbation strength $\delta_{\max}$ that ensures that the transfer function remains bounded, $||T_\delta||<\delta^{-1}$ for $\delta<\delta_{\max}$.
\end{abstract}

\section{INTRODUCTION}

Quantum control has become hugely popular in recent years due to the promise of novel technologies exploiting quantum phenomena. However, technological applications usually require robustness to bring quantum technologies out of the laboratory and into real-world application spaces. So it is crucial to find ways to assess robustness and engineer robust quantum controls. For classical robust control, many tools have been developed, including structured singular value analysis~\cite{Zhou}, to calculate stability margins and assess performance robustness in the presence of perturbations. However, it is often difficult to apply these tools to quantum systems for reasons ranging from the non-commutativity of state variables in, e.g., quantum optics~\cite{IRP_quantum_survey_as_published}, to the lack of stability margins.

Ideal quantum systems are governed by Hamiltonian dynamics. Quantum control problems are typically formulated as open-loop, model-based control problems, controlling Hamiltonian dynamics by changing the total Hamiltonian that governs the evolution of the system via the application of control fields. However, it is possible to formulate quantum control problems as feedback control problems~\cite{Edmond_IEEE_AC}. In general this results in a non-autonomous control system, but under certain conditions we can formulate quantum control as a time-invariant feedback control problem~\cite{bilinear_constant_input}. Such problems are in theory amenable to the application of classical robust control tools, but there are still numerous obstacles. Most prominently, Hamiltonian quantum systems are inherently unstable as the poles of the transfer function are entirely on the imaginary axis, and coherent \emph{Hamiltonian} control, not to be confused with coherent \emph{feedback} control, often employed in optical networks, cannot change this.

The situation changes for open quantum systems, as decoherence acts as a stabilizing controller for quantum systems. Open quantum systems are therefore in principle amenable to the application of tools like structured singular value analysis. Motivated by recent results~\cite{robust_performance_open}, we specifically consider controlled quantum systems with phase decoherence. With decoherence providing stability margins, we consider the effect of structured perturbations of the Hamiltonian on the performance. Despite the apparent similarity of the problem formulation to classical robust control problems, significant differences quickly become obvious. Structured perturbation of the Hamiltonian must respect Hermitian symmetry to be physical.

In addition, linear perturbations to the Hamiltonian affect system decoherence and give rise to \emph{nonlinear} structured perturbation of the dynamics. Nonetheless, we can define a structurally perturbed transfer function, $T_\delta(s)$, evaluate its norm and find the maximum perturbation magnitude $\delta_{\max}$ that ensures $||T_\delta||_\infty \leq \delta^{-1}$ for all $\delta<\delta_{\max}$.
Finally, we compare the robustness of the performance in terms of transfer fidelity for a controlled ring.

\section{THEORY \& METHODS}
\subsection{Bloch Equation Formulation of Quantum Control}

The starting point for quantum control is often the Schr\"odinger equation
\begin{equation}
 \imath \hbar \tfrac{d}{dt} \ket{\Psi(t)} = H \ket{\Psi(t)},
\end{equation}
where $H$ is the Hamiltonian and $\ket{\Psi(t)}$ denotes a wavefunction representing the quantum state of the system. Control in this context is typically formulated as replacing the system's natural Hamiltonian $H_0$ by a controlled Hamiltonian, $H = H_0 + \sum_m f_m(t) H_m$. In this formulation the Hamiltonian depends linearly on the controls, resulting in a bilinear control problem. Here, however, we will restrict ourselves to time-invariant forcing terms $f_m$. The Schr\"odinger equation formulation is limited to closed systems and pure quantum states but it can easily be generalized by replacing the wavefunctions by density operators $\rho$ acting on the Hilbert space $\H$. For any $\H$ with dimension $N<\infty$ we can choose an orthonormal basis for the Hermitian operators on $\H$ and expand both the state $\rho$ and Hamiltonian $H$ with regard to this basis, resulting in a convenient, real representation of the dynamics,
\begin{equation}
  \tfrac{d}{dt} \vec{r}(t) = \A_H \vec{r}(t),
\end{equation}
where $\vec{r} \in \RR^{N^2}$ and $\A_H \in \RR^{N^2 \times N^2}$. It is easy to show that $\A_H$ is an adjoint representation of the system. Given $H=H_0+\sum_m f_m(t) H_m$, we have $\A_H = \A_0 + \sum_m f_m(t) \A_m$. The generators $\A_m$ for $m=0,1,\cdots$, corresponding to the Hamiltonian dynamics, are real anti-symmetric matrices (rotation generators) with purely imaginary eigenvalues, but this changes when decoherence is added.

\subsection{Decoherence and Stability Margins}

We consider systems subject to decoherence in the form of dephasing in the Hamiltonian basis, a common type of decoherence observed, e.g., when electrons or nuclear spins in a magnetic field precess at (slightly) different rates due to local field inhomogeneities, resulting in them getting out of phase and the ensemble losing coherence. Mathematically, the process can be modeled by
\begin{equation*} \textstyle
  \tfrac{d}{dt}\rho = -\imath [H,\rho]
   + \gamma \sum_k \left( V_k\rho V_k - \tfrac{1}{2}\{V_k^2,\rho\} \right),
   \quad [H,V_k]=0,
\end{equation*}
where the $V_k$'s are Hermitian ``decoherence" operators, and $[\cdot,\cdot]$ and $\{\cdot,\cdot\}$ denote the commutator and anti-commutator, resp. These systems have a number of stable steady states, determined by the interplay of the Hamiltonian and the decoherence. We are interested in their robustness under Hamiltonian uncertainty. In previous work~\cite{statistical_control}, we studied robustness for small perturbations using measures such as the logarithmic sensitivity. Here, we aim to assess the robustness to larger perturbations of the Hamiltonian, which should be physical, i.e., at a minimum, preserve Hermitian symmetry. If decoherence acts in the Hamiltonian basis, changes to the Hamiltonian also indirectly affect the dephasing generators $V_k$, often in a non-linear fashion, resulting in complicated structures for the perturbations of interest.

\subsection{``Dephasing-Structured'' Perturbations}

Let $H$ be the nominal Hamiltonian and $\tilde{H} = H + \Delta H$ a perturbed Hamiltonian, linear additive in physically meaningful parameters. Express the Hamiltonian and decoherence operators in the eigendecomposition of $H$,
\begin{equation} \textstyle
 H = \sum_n \lambda_n \Pi_n, \quad V_k = \sum_n c_{k,n} \Pi_n,
\end{equation}
where $\Pi_n$ is the projector on the eigenspace associated with the eigenvalue $\lambda_n$ and $\{c_{k,n}\}_{n=1}^N$ are eigenvalues of the decoherence operator $V_k$, which must satisfy certain conditions to ensure the decoherence processes remain physical~\cite{Sophie_PRL_decoherence}.

Expressing the decoherence operators in terms of the eigenspaces of $H$ secures $[H,V_k]=0$, that is, the decoherence is pure dephasing acting in the Hamiltonian basis. Similarly, for the perturbed system,
\begin{equation} \textstyle
 \tilde{H}  = \sum_n \tilde{\lambda}_n \tilde{\Pi}_n, \quad \tilde{V}_k = \sum_n c_{k,n} \tilde{\Pi}_n.
\end{equation}
The perturbation of the Hamiltonian changes both the eigenvalues and eigenspaces, including the eigenspaces of the decoherence operators, but not the decoherence rates, determined by the coefficients $c_{k,n}$ in our model. For decoherence to act in the Hamiltonian basis of the perturbed system, we need $[\tilde{H},\tilde{V}_k]=0$, which implies that a linear perturbation of $H$ leads to a non-linear perturbation of the decoherence terms. Consequently, $\tilde{\A}$ (the perturbed $\A$) is nonlinear in $\Delta H$, a departure from the classical paradigm, considering additive perturbations written as $\delta \S$, $\delta$ being its size and $\S$ its $\delta$-independent structure. Having chosen a suitable operator basis for the Hilbert space, with respect to which to expand the operators $H$, $V_k$, $\tilde{H}$, $\tilde{V}_k$, we obtain the real $N^2\times N^2$ dynamical generators
\begin{equation} \textstyle
 \A = \A_H + \sum_k \A_{V_k}, \quad
 \tilde{\A} = \tilde{\A}_H + \sum_k \tilde{\A}_{V_k} = \A + \delta \S(\delta),
\end{equation}
with the ``dephasing-structured'' perturbation $\S(\delta) = \delta^{-1} (\tilde{\A}(\delta)-\A)$ that has an unusual $\delta$ dependency. Even for simple Hamiltonian perturbations, the resulting perturbation matrices $\S(\delta)$ have a complicated structure that cannot be reduced to a simple block-diagonal structure with blocks consisting of real or complex diagonal matrices or general complex blocks. Therefore, standard tools available for structured singular value analysis such as the \texttt{mussv} function in \texttt{MATLAB} cannot be applied.

\subsection{Dephasing-Structured Perturbation Analysis}

We consider a nominal $\A$ and a perturbed $\tilde{\A}$-dynamics and assess the difference \emph{relative to the perturbed dynamics:}
\begin{align}\label{eq:T}
 T_\delta(s)
 &=\left[(sI-\tilde{\A}(\delta))^{-1}-(sI-\A)^{-1}\right]
   \left[(sI-\tilde{\A}(\delta))^{-1}\right]^{-1}\nonumber\\
 &=\delta(sI-\A)^{-1}\S(\delta).
\end{align}
Scaling relative to the perturbed dynamics appears counter-intuitive but has some advantages~\cite[Eqs.\ (2.41), (2.42)]{Safonov_Laub_Hartmann} as the perturbed dynamics are the true dynamics, and also simplifies the frequency sweep.

The transfer matrix $T_\delta(s)$ and the behavior of $\|T_\delta (\imath \omega)\|$ versus $\delta$ is the central object of concern. Calculating the norm of the transfer function is complicated by the fact that $T_\delta(\imath\omega) = \delta(\imath\omega\I -\A)^{-1} \S(\delta)$ is ill-defined for $\omega=0$ as the $N^2 \times N^2$ matrix $\A$ for a system subject to pure dephasing always has $N$ zero eigenvalues~\cite{robust_performance_open} corresponding to constants of motion. We can solve this problem by eliminating the null-space of $\A$ corresponding to the constants of motion (trace invariants) and replacing $(\imath\omega \I-\A)^{-1}$ by a suitable, effective inverse~\cite{robust_performance_open}. Furthermore, the norm of the transfer function $T_\delta(\imath\omega)$ for a quantum system subject to pure dephasing in the Hamiltonian basis assumes its maximum when $\omega$ is an eigenfrequency of the total Hamiltonian of the system, including the effects of the control and perturbations, as exemplified in Fig.~\ref{fig1}(a). This means it suffices to calculate the eigenfrequencies of the system, which correspond to the energy differences between quantum states that make up the eigenbasis of the system, $\Delta E = \hbar\omega$. Since there are at most $N(N-1)/2$ distinct eigenvalues ($\pm$ pairs), we only have to evaluate the transfer function for this discrete set of frequencies to obtain the maximum, which greatly accelerates the computation.

\noindent{\bf Fundamental objective:} The fundamental objective is to calculate the maximum perturbation strength $\delta_{\max}$ such that $\|T_\delta\|_\infty \leq \delta^{-1}$ for all $\delta < \delta_{\max}$. If $\S$ is independent of $\delta$, it is easily seen that $\delta_{\mathrm{max}}=1/\sup_\omega \mu(G(\imath \omega))$, where $\mu$ is the structured singular value of $G(s)$, the $2 \times 2$ block connection matrix around which the diagonally perturbed and fictitious feedback matrices are wrapped~\cite[Fig.\ 10.5, Th.\ 10.8]{Zhou}.

\noindent{\bf Computational solution:} We choose a regular grid for $\ld\delta =\log_{10} \delta$ and numerically calculate the norm of the transfer function $T_\delta(\imath\omega)$ as described above. We then find $\delta_{\max} = f(\delta_{\max})$ for $f(\delta)=\|T_\delta\|^{-1}$ by fitting $f(\delta)$ and calculating the intersection point $\delta_{\max}=f(\delta_{\max})$ \emph{if such a fixed point exists}. When $f(\delta)$ is continuous and follows a power law dependence, a linear fit $y=ax+b$ of $\ld f(\delta)$ vs $\delta$ is performed for $x=\ld \delta$ and $y=\ld f(\delta)$ and the intersection point is calculated as $\delta_{\max}=10^{-b/(a-1)}$. Otherwise, a spline fit is performed and the intersection point computed numerically using the \texttt{MATLAB} function \texttt{fzero}, which uses a combination of bisection, secant, and inverse quadratic interpolation methods. There is a complication, however, as the perturbation $\S$ of the dynamics in the Bloch representation has itself a dependency on $\delta$ as a linear Hamiltonian perturbation $\Delta H$ causes a non-linear perturbation in the decoherence part of the dynamics if decoherence acts in the Hamiltonian basis. However, it may be expected that for sufficiently small $\delta$, $\S(\delta)$ is approximately constant. Occasional outliers in the data are removed prior to fitting using \texttt{MATLAB}'s \texttt{rmoutliers}. By default outliers are values of more than three scaled median absolute deviations.

\noindent{\bf The caveat:} Continuity of $f(\cdot)$ can be traced back to continuity of the decoherence operator $V_k(\delta)$ solution to the pure dephasing condition $[H(\delta),V_k(\delta)]=0$. Write the solution as $V_k=(v_{k,1}, v_{k,2},\cdots,v_{k,N})$ and define $\vv_k=(v_{k,1}^T, v_{k,2}^T,\cdots,v_{k,N}^T)^T$. Then the pure dephasing condition in the adjoint representation can be written as $\left( I \otimes H(\delta) - H(\delta) \otimes I\right) \vv_k(\delta) =0$. The question is whether a continuous basis in the kernel of $\mathrm{ad}_{H(\delta)}$ exists. For symmetric rings at the nominal $\delta=0$ the dimension of the null space of $\mbox{ad}_{H(\delta)}$ changes, so that Dole{\u z}al's theorem~\cite{Dolezal1} does not apply. However, a generalization~\cite{Dolezal2} implies:
\begin{theorem}
 There are analytic branches $\vv_k(\delta)$ in solutions to $\mathrm{ad}_{H(\delta)} \vv_k(\delta)=0$ in a neighborhood of $\delta=0$. Moreover, $\delta\S(\delta)$ is analytic in $\delta$ in the same neighborhood. Finally, $T_\delta(\imath \omega)$ is real-analytic in both $\delta$ and $\omega$.
\end{theorem}
\begin{proof}
The existence of analytic branches $\vv_k(\delta)$ in the kernel of $\mathrm{ad}_{H(\delta)}$, even under varying dimension with $\delta$, is guaranteed by~\cite{Dolezal2}, as a corollary of the Weirstrass factorization theorem~\cite[Chap. 1]{KnoppII}. Hence, the operator $V_k(\delta)$ is analytic and since an analytic function of an analytic function is analytic (Fa\`a di Bruno expansion~\cite{primer_analytic}), the Lindbladian $\left(\tilde{V}_k(\delta)\rho \tilde{V}_k(\delta)-\tfrac{1}{2}\{\tilde{V}_k^2(\delta),\rho\}\right)$ is analytic in $\delta$.

The transcription from the Lindblad to the Bloch formulation $(\tilde{H},\tilde{V}_k) \mapsto \tilde{\A}$ is analytic in both arguments, so $\tilde{\A}(\delta)$ and $\delta \S(\delta)$ are analytic. Analyticity of $T_\delta (\imath \omega)$ follows trivially.
\end{proof}

\subsection{Application to Spin Systems}

We apply this analysis to coupled spin systems subject to energy landscape control considered in previous work~\cite{chains_QINP, rings_QINP}. While applicable to spin systems in general, as a concrete example we consider rings and chains with dynamics restricted to the single excitation subspace~\cite{chains_QINP,rings_QINP}. Assuming nearest-neighbor coupling between spins, the Hamiltonians take on a simple symmetric tridiagonal (chains) or cyclic (rings) structure. The off-diagonal elements are determined by the intrinsic interaction strength between adjacent spins. The diagonal elements are mostly determined by the externally controlled energy landscape. Here, the energy landscape controls are simply constants added to the diagonal corresponding to local, static potentials on the spins. We can define physically meaningful structured perturbations.

For a chain of $N$ uniformly coupled spins with XX coupling we have
\begin{subequations}
\begin{align}
 H_{\rm Ch} &= \textstyle\sum_{n=1}^{N-1} e_{n,n+1} + e_{n+1,n}, \\
 H_{\rm Ct} &= \textstyle\sum_{n=1}^{N} D_n e_{nn},
\end{align}
\end{subequations}
and for a ring, $H_{\rm ring} = H_{\rm chain} + e_{1N}+e_{N1}$, where $e_{mn}$ is an $N\times N$ matrix with $1$ in the $m,n$ position and zeros otherwise. $D_n$ are the control parameters. There are $N-1$ basic structured perturbations of individual couplings between spins
\begin{equation}
 S_n^H = e_{n,n+1} + e_{n+1,n}, \quad n=1, \cdots, N-1
\end{equation}
and for rings, $S_N^S=e_{N,1}+e_{1,N}$, and $N$ perturbations for the controlled elements on the diagonal
\begin{equation}
 S_n^D = e_{n,n} \quad n=1,\cdots, N
\end{equation}
where the superscript $H$ indicates that the uncertainty pertains to the system Hamiltonian and $D$ indicates uncertainty in the controller.

\section{RESULTS AND DISCUSSION}
\subsection{Preliminary Results}

\begin{figure}
  \subfloat[Norm $\|T_\delta(\imath\omega)\|$ vs. $\omega$, $\delta$] {\includegraphics[width=0.5\columnwidth]{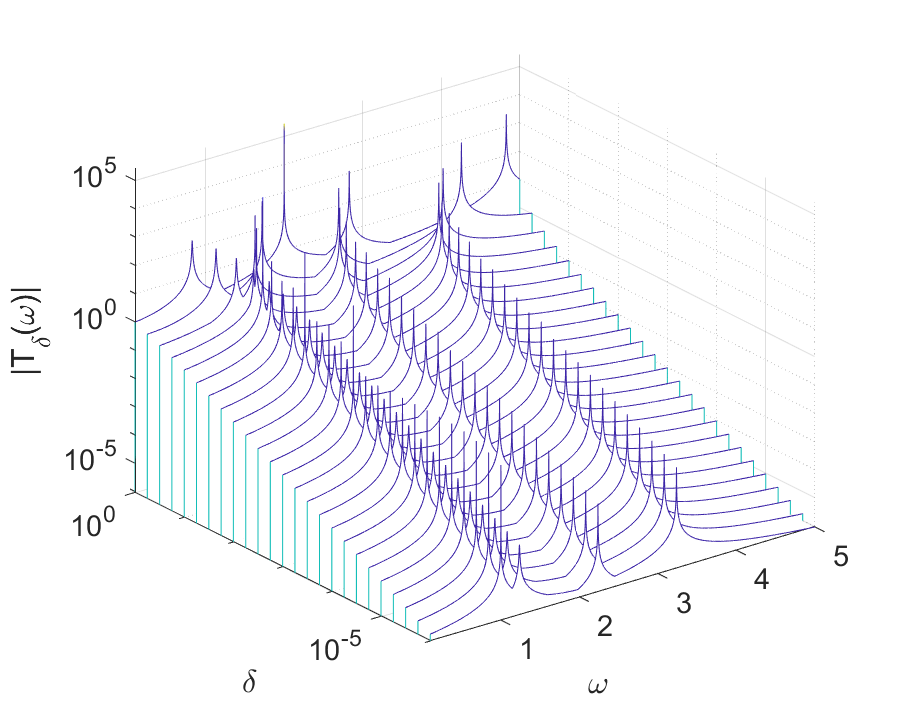}}
  %\subfloat[Norm $||T_\delta||$ of transfer function vs. $\delta$]
  \subfloat[Norm $||T_\delta||$ vs. $\delta$]
  {\includegraphics[width=0.5\columnwidth]{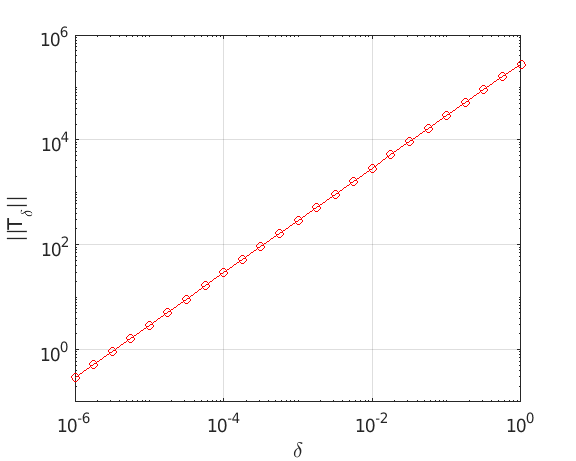}}
  \caption{Norm vs. (a) $\delta, \omega$ and (b) $\delta$ of transfer matrix for uncontrolled $N=4$ chain. (a) shows that $\|T_\delta(\imath\omega)\|$ assumes its maximum at distinct spikes corresponding to eigenfrequencies of the system. (b) shows the norm calculated by searching over eigenfrequencies.}\label{fig1}
\end{figure}

Fig.~\ref{fig1}(a) shows that $\|T_\delta(\imath\omega)\|$ assumes its maximum when $\omega$ is an eigenfrequency of the system and that it suffices to evaluate $\|T_\delta(\imath\omega)\|$ for the eigenfrequencies of the system. The observed four distinct peaks seen around $\omega \in \{1,1.23, 2.23, 3.23\}$ are consistent with expectations. For the uncontrolled, unperturbed chain it is easy to verify by direct calculation that the Hamiltonian has four distinct eigenvalues at $1,\sqrt{5}\pm 1,\sqrt{5}$, two of which are ($1$, $\sqrt{5}$) two-fold degenerate. We find empirically that if the Hamiltonian is perturbed the eigenfrequencies shift.  Fig.~\ref{fig1}(b) further suggests a scaling of $||T_\delta||$ consistent with theoretical expectations, following a power law $||T_\delta|| \propto \delta^\alpha$. A numerical fit of the data gives $\alpha = 0.9981$ with $95$\% confidence interval $(0.9971, 0.9992)$, i.e., slightly below $1$, the value expected if $\S$ was independent of $\delta$.

\begin{figure}%[thpb]
 \subfloat[Critical frequencies vs $\delta$] {\includegraphics[width=0.5\columnwidth]{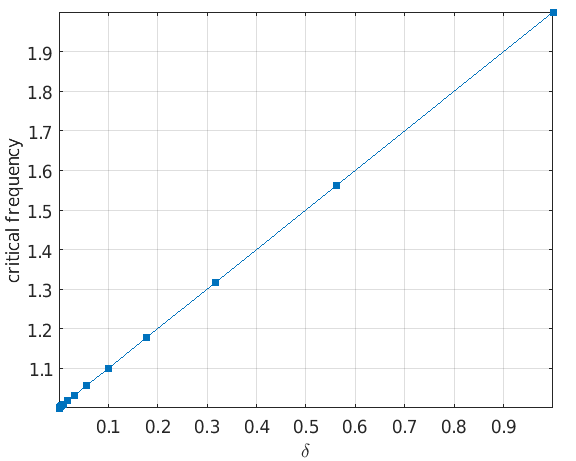}}
 \subfloat[Poles of transfer matrix]
 {\includegraphics[width=0.5\columnwidth]{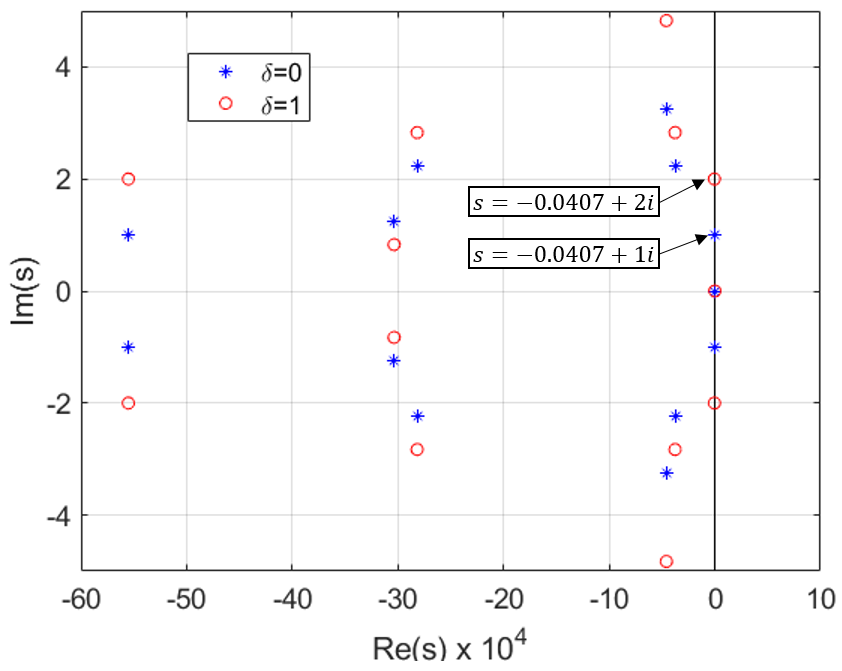}}
 \caption{(a) shows a linear increase of the critical frequency from $1$ for $\delta=0$ to $2$ for $\delta=1$. The critical frequency can be explained by the frequency of pole closest to $x=0$ in subplot (b), around $1$ for $\delta=0$, increasing to $2$ for $\delta=1$. (b) also shows that the Hamiltonian perturbation only shifts the frequencies of the poles along the imaginary axis.}\label{fig2}
\end{figure}

Subplot (a) in Fig.~\ref{fig2} shows the critical frequency $\omega_*$ for which $\|T_\delta(\imath\omega)\|$ assumes its maximum, which can be explained by plotting the poles of the transfer matrix, as shown in subplot (b): the critical frequency observed between $1$ ($\delta=0$) and $2$ ($\delta=1$) corresponds to the poles with the smallest real part. Crucially, subplot (b) shows that the Hamiltonian perturbation only moves the poles along the vertical axis.
Even if the basis in which decoherence acts changes as a result of the Hamiltonian perturbation, provided the decoherence rate $\gamma$ itself is not affected, the real part of the poles does not change. Therefore, no amount of Hamiltonian perturbation destabilizes the system as the poles remain in the left half-plane.

\subsection{Uncontrolled Chains and Rings}

\begin{figure}
 \centering
 \includegraphics[width=\columnwidth]{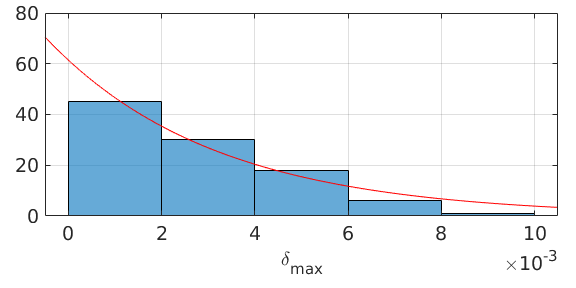}
 \caption{Typical exponential distribution of $\delta_{\max}$ over $100$ decoherence processes with a fixed strength $\gamma$. The example shown: $N=4$ ring, $S_1$ perturbation.}\label{fig3}
\end{figure}

\begin{figure}%[thpb]
 \subfloat[Chain $N=4$: Power law fit]{\includegraphics[width=0.5\columnwidth]{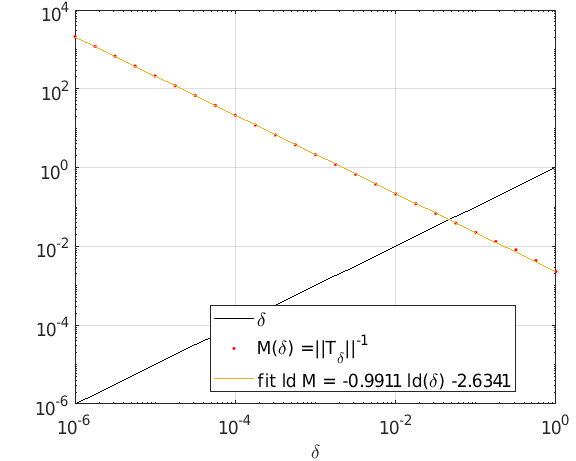}}
 \subfloat[Ring $N=4$: Generic spline fit] {\includegraphics[width=0.5\columnwidth]{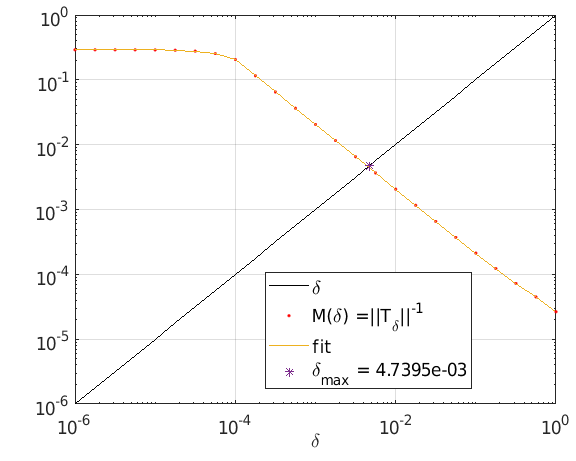}}
 \caption{Determination of $\delta_{\max}$ by (a) performing a linear fit of $y=\log_{10} f(\delta)$ for $f(\delta)=||T_\delta||^{-1}$ vs $x=\log_{10} \delta$ and (b) performing a spline fit of $f(x)$, and calculating the intersection with $y=x$.
 Observe the break-point in plot (b) indicating that $\S$ depends on $\delta$.}\label{fig4}
\end{figure}

We first consider one of the simplest cases, uncontrolled rings and chains. To systematically assess the effect of Hamiltonian perturbations on uncontrolled chains and rings we choose $N=4$. For the chain, there are three fundamental structured perturbations, $S_1$ to $S_3$, corresponding to perturbations to the couplings between spins $(1,2)$, $(2,3)$ and $(3,4)$, respectively. For rings, there is an additional perturbation corresponding to $(4,1)$. However, due to the symmetry with respect to cyclic permutations, it suffices to consider a single structured perturbation, e.g., $S_1$, for rings. For chains we explore all three perturbations, although we expect the effect of $S_1$ and $S_3$ to be identical due to inversion symmetry.

The effect of dephasing with unknown decoherence rates was modeled by selecting $100$ decoherence processes corresponding to valid pure dephasing processes, which were generated by a low-discrepancy sampling of the decoherence rate parameter space and eliminating unphysical combinations of parameters~\cite{Sophie_PRL_decoherence}. Each decoherence process can further be scaled to vary the total decoherence strength $\gamma$. For each structured perturbation to the Hamiltonian, $S_n$, $\delta_{\max}$ was calculated for all $100$ decoherence processes, scaled to a fixed decoherence strength $\gamma$. The Lilliefors test, a normality test based on the Kolmogorov–Smirnov test, rejected the null hypothesis that the resulting distributions for $\delta_{\max}$ come from the normal family at the $0.001$ significance level for both rings and chains, and all structured perturbations, and the distributions appear exponential as shown in Fig.~\ref{fig3}.

In determining $\delta_{\max}$ we observe an interesting difference between chains and rings. For all the uncontrolled chains, $||T_\delta||^{-1}$ fitted a power law over a wide range of $\delta$ as shown in Fig.~\ref{fig4}(a), while for the rings we consistently observed a cut-off value for $\delta$, below which $||T_\delta||$ appears to be effectively constant, as shown in Fig.~\ref{fig4}(b). Despite this difference, Fig.~\ref{fig5} shows that $\delta_{\max}$ as a function of $\gamma$ appears to follow a similar power law scaling for both uncontrolled rings and chains of size $N=4$. For the chain, we compared the distributions for $S_1$, $S_2$ and $S_3$ perturbations, and, as expected, the $S_1$ and $S_3$ distributions were identical (Pearson correlation coefficient of $1.0000$). The distributions for $S_1$ and $S_2$ were still strongly correlated with a (Pearson) correlation coefficient of $0.9574$, suggesting that for simple systems like uncontrolled rings or chains, it may be sufficient to consider a single perturbation.

\begin{figure}
\centering
\includegraphics[width=\columnwidth]{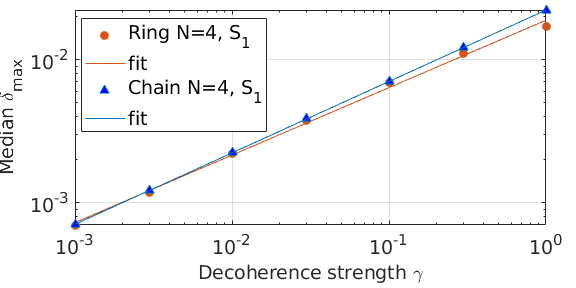}
\caption{Median of $\delta_{\max}$ over $100$ decoherence processes scaled by an overall strength $\gamma$ vs $\gamma$ for uncontrolled rings and chains of size $N=4$, respectively, suggests a power law scaling of $\gamma^\alpha$ with $\alpha$ just below $0.5$ for both cases.}\label{fig5}
\end{figure}

\subsection{Controlled Spin Networks}

\begin{figure*}
\centering
 \includegraphics[width=\textwidth]{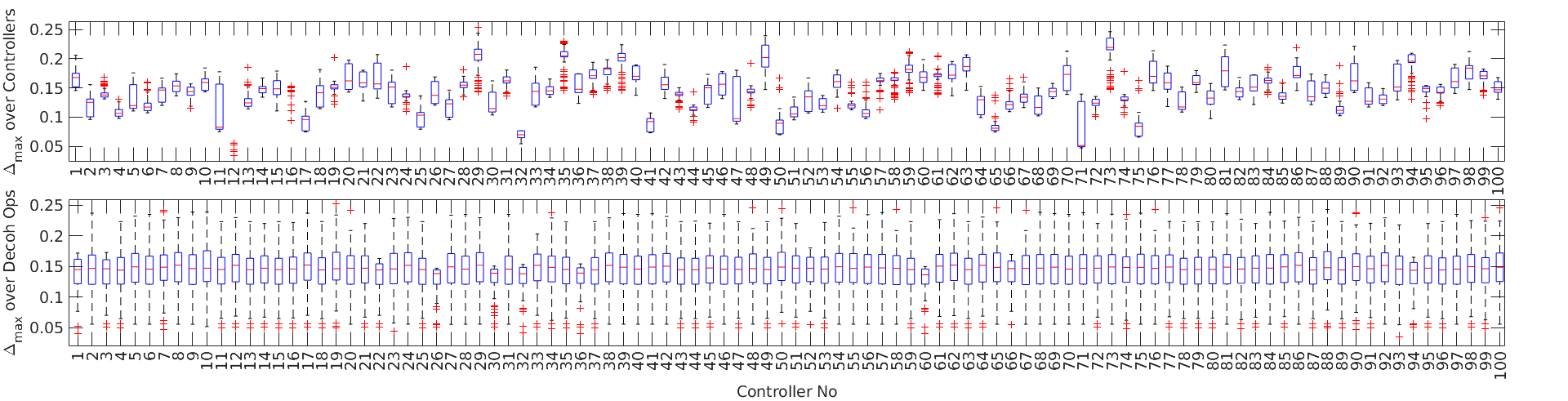}
 \caption{Box plot of $\delta_{\max}$ distributions for controlled ring ($N=5$). Distributions over the decoherence processes for different controllers (bottom) appear similar; distributions over the controllers for fixed coherence processes (top) differ wildly.}\label{fig:BoxPlot}
\end{figure*}

Next, we consider simple spin networks subject to energy landscape control, which in this context simply means that the diagonal elements of the Hamiltonian are non-zero. The controllers considered have been optimized to achieve high-fidelity information transfer between an input and output node, here $1$ and $3$, respectively. It was shown in previous work~\cite{Edmond_IEEE_AC} that these controllers have interesting robustness properties in that the differential sensitivity of the transfer fidelity for superoptimal controllers, i.e., controllers that achieve unit-fidelity transfer, vanishes, which runs counter to the trade-off between performance and robustness that is commonly seen for classical systems~\cite{Safonov_Laub_Hartmann}. However, other work indicates that this performance advantage disappears in the presence of decoherence~\cite{soneil_mu, CDC_decoherence}. Moreover, differential sensitivity gives no information how the system responds to larger perturbations over a prolonged period of time, or what the critical frequencies are.

Fig.~\ref{fig:BoxPlot} shows the \textbf{distributions for $\delta_{\max}$} over different decoherence processes and different controllers. There is little variation in the distributions over the decoherence for the different controllers but large variation in the distributions over different controllers for different decoherence operators. This is not too surprising, as $\delta_{\max}^{-1}$, can be regarded as a measure of the gain (or damping) of a perturbation, primarily determined by the decoherence rates.

\noindent{\textbf{Robust Performance.}} While the gain of a potential disturbance is a useful measure of robustness, it does not give any direct insight into the effect of the perturbation on the performance of a controller in terms of the desired time-domain information transfer. The performance measure used here is the overlap of the state $\rho(t_f)$ of the system at a fixed time $t_f$ with a desired state $\rho_{\rm out}$, $\mathcal{F} = \Tr[\rho_{\rm out} \rho(t_f)]$. If $\rho_{\rm out}$ represents a pure target state then the maximum fidelity $\mathcal{F}$ is $\Tr(\rho_{\rm out}^2)=1$, assumed for $\rho(t_f)=\rho_{\rm out}$. Therefore, $1-\mathcal{F}$ is a measure of the transfer error. Fig.~\ref{fig:performance} illustrates the distributions of the transfer error for different levels of Hamiltonian perturbations $\delta$ and decoherence, quantified by the decoherence strength $\gamma$. Absent decoherence (Fig.~\ref{fig:performance}(a)), the transfer error increases with $\delta$ and the logarithm scale plot suggests a power law dependence for the median transfer error. When decoherence is active, even at the lowest level of $\gamma=0.001$ (subplot (b)), the performance is mainly limited by decoherence and the effect of Hamiltonian perturbations only becomes significant for perturbations several orders of magnitude larger than the decoherence. For example, for $\gamma=0.001$, the effect of the Hamiltonian perturbation on the median of the error distribution only becomes significant for $\delta=0.1$, and for $\gamma=0.1$, even Hamiltonian perturbations at the $\delta=0.1$-level barely increase the error (subplot (d)).

The results are not unexpected. While decoherence stabilizes the system and dampens the effect of perturbations, which is desirable, further analysis shows that the states stabilized by decoherence in this example are classical mixed states, which perform poorly for information transfer.

\begin{figure}
 \subfloat[Transfer error $\gamma=0$]{\includegraphics[width=.5\columnwidth]{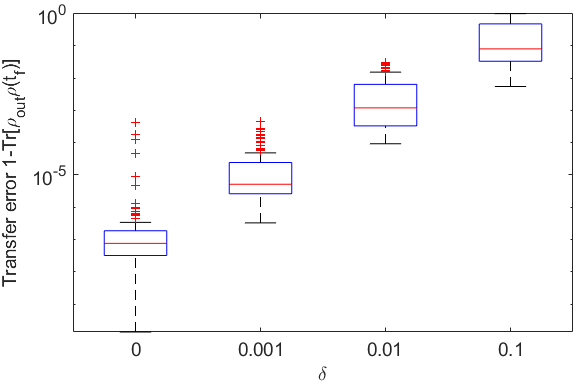}}
 \subfloat[Transfer error $\gamma=0.001$]{\includegraphics[width=.5\columnwidth]{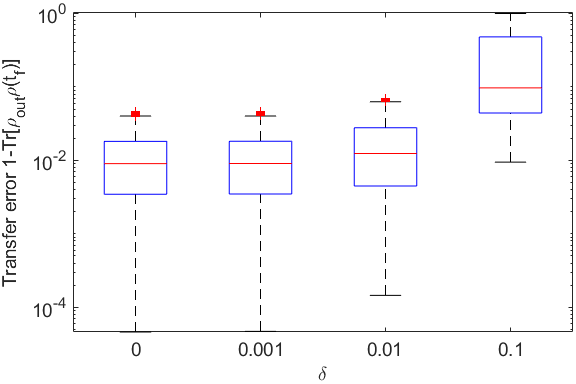}}\\
 \subfloat[Transfer error $\gamma=0.01$]{\includegraphics[width=.5\columnwidth]{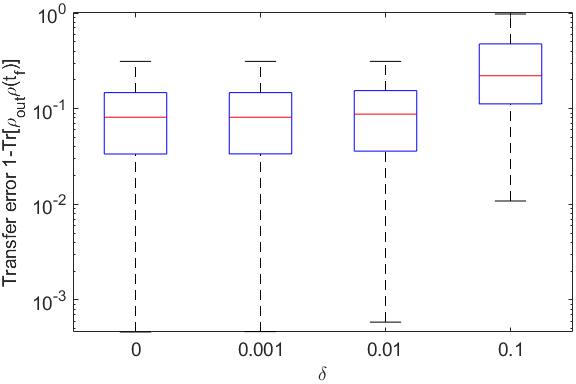}}
 \subfloat[Transfer error $\gamma=0.1$]{\includegraphics[width=.5\columnwidth]{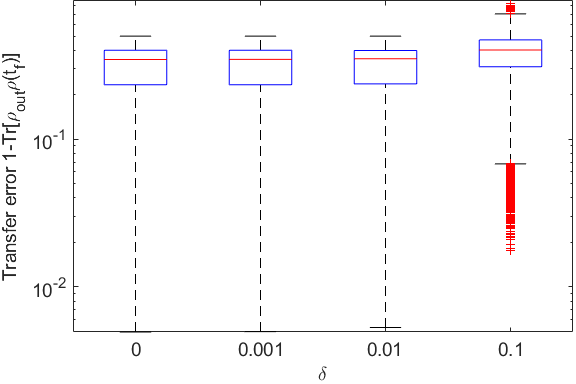}}
 \caption{Transfer error for various decoherence rates.}\label{fig:performance}
\end{figure}

The distributions for $\delta_{\max}$ over the decoherence processes for fixed controllers appear closer to normal distributions---the Lilliefors test rejected the null hypothesis that the distributions came from a normal family at the $0.05$ significance level only for a small fraction ($<15$\%) of the controllers considered. The distributions of $\delta_{\max}$ over different controllers for a fixed decoherence process, on the other hand, follow no clearly discernible pattern and do not appear normal.

\section{CONCLUSIONS AND FUTURE WORK}

A novel way to study robustness for quantum systems, introducing a substitute for $\mu$-analysis, referred to as $\delta_{\mathrm{max}}$-analysis, was presented. The transfer function analysis shows that a Hamiltonian system can never be completely robust in that perturbations of the system at its critical frequencies can become unbounded, which is bad news for applying established classical methods to robust quantum control. Decoherence changes this, providing natural damping that removes the purely imaginary poles in the transfer function, limiting the gain of any disturbance. Thus, generic decoherence stabilizes the system, making it more robust to disturbances. Moreover, a quantum system stabilized by decoherence cannot be destabilized by Hamiltonian perturbations even if the perturbation changes the basis in which decoherence acts if the overall decoherence rates are unchanged. Decoherence thus imbues quantum systems with significant robustness. However, states robustly stabilized by decoherence are often classical or at least do not offer any quantum advantage. A general theory of robust quantum control is needed to understand the trade-off between quantum advantage and robustness, including new tools to analyze the robustness of closed quantum systems (i.e., those without dissipation and dephasing), where a stability margin cannot be defined in the same way as for classical systems.

For the spin systems studied, the gains of the controlled chains (\emph{cf.}\ Fig.~\ref{fig4}) appear substantially lower than for the uncontrolled chains. This ultimately requires a physical explanation, and apparent scaling laws for the transfer function norm (again \emph{cf.} Fig.~\ref{fig4}) should be investigated rigorously in the future. Further work is also needed to understand how the fidelity of these spin transfer systems changes as $\delta$ increases, and how, if at all, this is encapsulated in the behavior of the critical frequencies, $\|T_\delta\|$ and $\mu$.

Finally, experimental confirmation of the results will be needed, which could be achieved by controlling the coupling between transmon qubits~\cite{Martinis} or mapping this Hamiltonian onto a system of cold, trapped atoms~\cite{coldAtomsXX}. Our theoretical method can also be extended to study similar Heisenberg-type Hamiltonians in, e.g. transmon qubits~\cite{transmon}, trapped ions~\cite{Monroe} or cold atoms in optical lattices~\cite{superexchange}. Applications to other systems (e.g.\ EIT, coupled qubits in cavities) should also be considered, as there are well-defined notions of quantumness (e.g. entanglement measures) that can be applied in these cases.

%\bibliographystyle{plain}
%\bibliography{biblio/edmond,biblio/physics}

\end{document}